\documentclass[letterpaper, 10 pt, conference]{ieeeconf} 
\IEEEoverridecommandlockouts  

\usepackage{amsmath}
\usepackage{amssymb}
\usepackage{amsthm}
\usepackage{mathtools}
\usepackage{derivative}
\NewDerivative{\tpdv}{\partial}[style-frac=\tfrac]
\usepackage{xcolor}
\usepackage{bm}

\usepackage[noadjust]{cite}
\usepackage{url}

\usepackage{graphicx}
\usepackage[export]{adjustbox} 
\graphicspath{{figures/}}

\newtheorem{theorem}{Theorem}
\newtheorem{lemma}{Lemma}

\newtheorem{corollary}{Corollary}

\theoremstyle{definition}
\newtheorem{definition}{Definition}

\newtheorem{remark}{Remark}

\newtheorem{property}{Property}

\newcommand{\R}{\mathbb{R}}

\newcommand{\T}{^\top}

\newcommand{\bzero}{\mathbf{0}}

\renewcommand{\bf}{\mathbf{f}} 
\newcommand{\bg}{\mathbf{g}}

\newcommand{\bk}{\mathbf{k}}

\newcommand{\bq}{\mathbf{q}}
\newcommand{\br}{\mathbf{r}}
\newcommand{\bs}{\mathbf{s}}

\newcommand{\bu}{\mathbf{u}}

\newcommand{\bx}{\mathbf{x}}

\newcommand{\bz}{\mathbf{z}}

\newcommand{\bC}{\mathbf{C}}

\newcommand{\bG}{\mathbf{G}}

\newcommand{\bI}{\mathbf{I}}

\newcommand{\bK}{\mathbf{K}}

\newcommand{\bM}{\mathbf{M}}

\newcommand{\bW}{\mathbf{W}}

\newcommand{\bY}{\mathbf{Y}}

\newcommand{\btheta}{\bm{\theta}}

\newcommand{\bpsi}{\bm{\psi}}

\newcommand{\bPhi}{\bm{\Phi}}

\newcommand{\bGamma}{\bm{\Gamma}}

\newcommand{\bLambda}{\bm{\Lambda}}
\newcommand{\bnu}{\bm{\nu}}

\newcommand{\bhat}[1]{\skew{3}{\hat}{\bm{#1}}}
\newcommand{\bdothat}[1]{\skew{3}{\dot}{\skew{3}{\hat}{\bm{#1}}}}
\newcommand{\btilde}[1]{\skew{3}{\tilde}{\bm{#1}}}
\newcommand{\bdottilde}[1]{\skew{3}{\dot}{\skew{3}{\tilde}{\bm{#1}}}}

\usepackage{xcolor}
\definecolor{myblue}{RGB}{49, 114, 174}
\definecolor{myred}{rgb}{0.796, 0.235, 0.2}
\definecolor{mygreen}{rgb}{0.22, 0.596, 0.149}
\definecolor{mypurple}{rgb}{0.584,0.345,0.698}

\usepackage{blindtext}

\title{\textbf{High Order Tuners for Adaptive Safety of Robotic Systems}}

\author{Mohammad Mirtaba and Max H. Cohen%
\thanks{Department of Electrical and Computer Engineering, North Carolina State University, Raleigh, NC, \texttt{\{mmirtab,mhcohen2\}@ncsu.edu}.}
}

\begin{document}
\maketitle
\begin{abstract}
    The combination of control barrier functions (CBFs) and adaptive control---a framework referred to as adaptive safety---has proven to be a powerful paradigm for safety-critical control of nonlinear systems with parametric uncertainties. Yet the theoretical conditions for forward invariance within this framework are often quite conservative, and may require using large adaptation gains to achieve acceptable performance, an approach that is traditionally discouraged in adaptive control. This paper mitigates these issues via high-order tuners, a recent class of higher-order adaptation laws that leverages different adaptation gains at different orders of differentiation. We illustrate that these high-order tuners decouple adaptation gain conditions from those placed on the initial conditions of the system required for set invariance. We extend these results to robotic systems whose linear-in-the-parameters structure proves particularly useful for adaptive control. The efficacy of our results are illustrated via simulations.
\end{abstract}

\section{Introduction}

The 
increasing
prevalence of autonomous systems
has raised concerns regarding their safety and the hazards they may pose to humans. 
Control barrier functions (CBFs) \cite{AmesTAC17}, which formalize safety via the notion of set invariance, have been successfully used to design controllers meeting safety objectives, such as constraining the position or velocity of robots. However, safety guarantees through CBFs are heavily based on model knowledge, and such guarantees may be violated in the case of model uncertainties. Hence, in realistic settings, unmodeled dynamics, parametric uncertainties, and unknown operating environments pose challenges to ensuring safe operation of autonomous systems using CBFs.

Adaptive control---a powerful control strategy that leverages online parameter estimation---has been extensively studied and applied to robotic systems with parametric uncertainties to meet control objectives such as stabilization or reference tracking \cite{SlotineIJRR87,SlotineLi,Krstic}. More recently, adaptive control has been used for safety by uniting online parameter estimation with techniques such as CBFs \cite{AndrewACC20,LopezLCSS21,L1CBF,IsalyACC21,Cohen}, model predictive control \cite{adetola2009adaptive,heirung2017dual}, and reference governors \cite{governercbf}.
Adaptive control and CBFs are often united using adaptive CBFs (aCBFs) \cite{AndrewACC20}, which build upon adaptive control Lyapunov functions \cite{Krstic} for adaptive stabilization of nonlinear systems.  
Adaptive CBFs ensure safety by maintaining the system state 
and parameter estimates within an augmented safe set in the joint state space of plant states and parameter estimates.
However, aCBFs impose conditions that 
imply forward invariance of every non-negative super level set of this safe set,
which can result in overly conservative behavior. To mitigate conservatism, robust aCBFs (RaCBFs) \cite{LopezLCSS21}
relax the aCBF invariance conditions, aligning them more closely with original CBF conditions \cite{AmesTAC17},
which
leads to less conservative performance. 
However, both aCBFs and RaCBFs require a restrictive condition for starting in the augmented safe set, 
potentially requiring the use of large adaptation gains, typically discouraged in adaptive control, to overcome initial parameter estimation errors to ensure safety.

High-order tuners \cite{hotadaptive,DanielLeACC22,TeelCDC22}
are a recent class of adaptation laws that draw inspiration from accelerated optimization algorithms, such as Nesterov's method \cite{nesterov269method,BoydJMLR16}, and have been shown to improve performance compared to traditional gradient-based adaptive update laws. While previous works have derived the update laws governing these high-order tuners from variational \cite{JordanPNAS16} and symplectic \cite{hotadaptive} principles, the resulting update laws typically consist of a traditional gradient-based adaptive update law in conjunction with a low-pass filter \cite{hotadaptive,TeelCDC22}. Although these methods have demonstrated success for parameter estimation \cite{hotadaptive,TeelCDC22} and adaptive stabilization \cite{DanielLeACC22} they have, to our knowledge, yet to be investigated in the context of adaptive safety.


In this paper, we integrate high-order tuners into the adaptive safety framework, illustrating how the resulting filtered update laws mitigate the challenges associated with requiring large adaptation gains for safety. This is achieved by pairing such update laws with a tunable RaCBF, where the CBF condition depends on the bandwidth of the filter, balancing robustness of the controller with the rate of adaptation. We first establish safety guarantees for control affine systems with matched parametric uncertainty before extending such results to robotic systems modeled using the manipulator equations. 
While recent approaches \cite{ZengLCSS24} have attempted
to leverage the linear-in-the-parameters (LIP) structure of the manipulator equations for adaptive safety, they inherit limitations reminiscent of those 
in adaptive control of robotic systems 
before the work of Slotine and Li \cite{SlotineIJRR87}. Namely, \cite{ZengLCSS24} requires invertibility of the estimated inertia matrix and access to acceleration measurements, both of which present practical challenges, whereas other approaches \cite{DeyLCSS25} do not exploit the LIP structure of the dynamics. In contrast, we leverage techniques based on \cite{TamasRAL22,CohenARC24} where safety is established by tracking a safe reference velocity. To accomplish this safe tracking objective, we propose a modification of the classic Slotine-Li controller \cite{SlotineIJRR87}, which yields bounded tracking (asymptotic tracking under additional assumptions) when paired with a high-order tuner, and allows for establishing safety guarantees when the reference velocity satisfies CBF-like conditions. Importantly, our approach does not require acceleration measurements nor invertibility of the estimated inertia matrix, and exploits the LIP dynamics. 

In summary, the contributions of this paper are threefold: \textbf{1)} we integrate high-order tuners with adaptive safety to mitigate restrictive adaptation gain conditions for control affine systems with matched parametric uncertainties; \textbf{2)} we propose a novel adaptive tracking controller for robotic systems with a high-order tuner adaptive update law; \textbf{3)} we integrate this tracking controller with our results on high-order tuners 
for safe adaptive control of robotic systems.



\section{Preliminaries and Problem Formulation}
\label{sec:prelim}

\noindent\textbf{Control Barrier Functions.}
Consider a nonlinear system:
\begin{align}
    \dot{\bx} = \bf(\bx) + \bG(\bx)\bu,
\label{eq:general_knowen_nonlinear_system}
\end{align}
with state $\bx \in \mathbb{R}^n$, control input $\bu \in \mathbb{R}^m$, drift dynamics $\bf:\mathbb{R}^n \rightarrow  \mathbb{R}^n$, 
and control directions $\bG:\mathbb{R}^n \rightarrow \mathbb{R}^{n\times m}$. 
Provided these dynamics are locally Lipschitz, then any locally Lipschitz feedback controller $\bk\,:\,\R^n\rightarrow\R^m$ produces the locally Lipschitz closed-loop system:
\begin{equation}\label{eq:closed-loop-dyn}
    \dot{\bx} = \bf_{\mathrm{cl}}(\bx) \coloneqq  \bf(\bx) + \bG(\bx)\bk(\bx).
\end{equation}
Hence, 
for any initial
condition \( \bx_0\), there exists a maximal interval of existence
\( I(\bx_0) := [0,\tau_{\max}) \subset \mathbb{R} \), where \( \tau_{\max} \in \mathbb{R}_{>0} \), such that
\( t \mapsto \bx(t) \) is the unique solution to (\ref{eq:closed-loop-dyn}) on \( I(\bx_0) \). 
In the CBF framework, the concept of safety is formalized using the notion of set invariance. Recall that a set $\mathcal{C}$ is forward invariant for \eqref{eq:closed-loop-dyn} if,
for any initial condition $\bx_0 \in \mathcal{C}$, the trajectory satisfies
$\bx(t)\in\mathcal{C}$ for all $t\in I(\bx_0)$.



When using CBFs to enforce safety properties, one typically considers 
sets defined as the zero superlevel set of a continuously differentiable function $h\,:\,\R^n\rightarrow\R$ as:
\begin{equation}
    \begin{aligned}
        \mathcal{C} \coloneqq & \{\bx \in \mathbb{R}^n \mid h(\bx) \ge 0\}.
        \label{eq:C}
\end{aligned}
\end{equation}
Barrier functions\footnote{See \cite{KondaLCSS21} for an in-depth account of the regularity conditions on barrier functions, including regular values and the domain where \eqref{eq:barrier} must hold.} are useful for verifying safety properties whereas CBFs help synthesize feedback controllers meeting such properties. Before formally defining CBFs, we define and state the main result with regard to barrier functions.

\begin{definition}[\cite{AmesTAC17,KondaLCSS21}]\label{def:barrier}
A continuously differentiable function $h\,:\,\R^n\rightarrow\R$ defining a set $\mathcal{C}$ as in \eqref{eq:C} is said to be a barrier function for the closed loop system \eqref{eq:closed-loop-dyn} if there exists\footnote{A continuous function $\alpha : \mathbb{R} \to \mathbb{R}$ is
extended class $\mathcal{K}_\infty$ ($\alpha\in \mathcal{K}_\infty^e$) if $\alpha$ is strictly increasing,
$\alpha(0)=0$, and $
\lim_{r \to \pm\infty} \alpha(r) = \pm\infty.
$} $\alpha\in\mathcal{K}_\infty^e$ such that for all $\bx\in\R^n$:
\begin{equation}\label{eq:barrier}
    \dot{h}(\bx) = \tpdv{h}{\bx}(\bx)\bf_{\rm{cl}}(\bx) \ge -\alpha\big(h(\bx)\big). 
\end{equation}
\end{definition}

\begin{theorem}[\cite{AmesTAC17,KondaLCSS21}]\label{thm:barrier}
    If $h$ is a barrier function for \eqref{eq:closed-loop-dyn}, then $\mathcal{C}$ is forward invariant for \eqref{eq:closed-loop-dyn}.
\end{theorem}

Finally, CBFs extend barrier functions to control systems \eqref{eq:general_knowen_nonlinear_system}, and provide conditions under which there exists inputs meeting the barrier condition in \eqref{eq:barrier}.

\begin{definition}[\cite{AmesTAC17}]
\label{def:cbf}
A continuously differentiable function $h\,:\,\R^n\rightarrow\R$ defining a set $\mathcal{C}$ as in \eqref{eq:C} is said to be a CBF for \eqref{eq:general_knowen_nonlinear_system} 
if there exists $\alpha\in\mathcal{K}_\infty^e$ such that for each $\bx\in\R^n$:
\begin{align}
\sup_{\bu \in \mathbb{R}^m}
\left\{
\tpdv{h}{\bx}(\bx)[\bf(\bx) + \bG(\bx)\bu]
\right\}
> -\alpha\big(h(\bx)\big).
\label{eq:cbf}
\end{align}
\end{definition}


The main utility of CBFs is their ability to instantiate \emph{safety filters}, a class of controllers that modify an existing controller in a minimally invasive fashion to ensure safety of the resulting closed-loop system. These controllers are typically computed by solving the quadratic program (QP):

\begin{equation}\label{eq:cbf-qp}
    \begin{aligned}
        \min_{\bu\in\R^m} \quad & \tfrac{1}{2}\|\bu - \bk_{\rm{d}}(\bx)\|^2 \\ 
        \mathrm{subject~to} \quad & \tfrac{\partial h}{\partial \bx}(\bx)[\bf(\bx) + \bG(\bx)\bu]
\geq -\alpha\big(h(\bx)\big),
    \end{aligned}
\end{equation}
where $\bk_{\rm{d}}\,:\,\R^n\rightarrow\R^m$ is a desired controller. When $h$ is a CBF meeting the strict inequality\footnote{See \cite[Remark 1]{CohenARC24} for further comments on strict vs nonstrict inequalities in the context of barrier functions and CBFs.} in \eqref{eq:cbf}, and all of the problem data is locally Lipschitz, the controller from solving \eqref{eq:cbf-qp} is also guaranteed to be locally Lipschitz \cite{CohenLCSS23,CohenARC24}.

\smallskip
\noindent\textbf{Adaptive Safety.}
We now review aCBFs \cite{AndrewACC20,LopezLCSS21}, which extend CBFs to systems with parametric uncertainty: 
\begin{align}
    \dot{ \bx} = \bf(\bx) + \bG(\bx)[\bu+\bm{\Phi}(\bx)\bm{\theta}],
    \label{eq:general_nonlinear_system}
\end{align}
with state $\bx \in \mathbb{R}^n$, control input $\bu \in \mathbb{R}^m$, unknown parameter vector $\bm{\theta} \in \mathbb{R}^p$, regressor matrix $\bm{\Phi}: \mathbb{R}^n \rightarrow \mathbb{R}^{m\times p}$, drift dynamics $\bf:\mathbb{R}^n \rightarrow  \mathbb{R}^n$, and control directions $\bG:\mathbb{R}^n \rightarrow \mathbb{R}^{n\times m}$. Adaptive CBFs unite techniques from adaptive control \cite{Krstic} with those from safety-critical control by pairing a safe controller $\bu=\bk(\bx,\bhat{\theta})$ based on a parameter estimate $\bhat{\theta}$ with a parameter update law $\odv{}{t}\bhat{\theta}=-\bGamma\bpsi(\bx)$ that dynamically adjusts these estimates online to ensure safety. Properties of this controller are then established by the studying the composite closed-loop system:
\begin{align}
\begin{bmatrix} \dot{\bx} \\ \bdothat{\theta} \end{bmatrix} 
&= \begin{bmatrix}
\bf(\bx) + \bG(\bx)\big(\bk(\bx,\bhat{\theta}) + \bm\Phi(\bx)\btheta \big) \\
-\bm\Gamma \, \bpsi(\bx)
\end{bmatrix},
\label{eq:augmented_adaptive_ss}
\end{align}
with an augmented state space of plant states $\bx$ and parameter estimates $\bhat{\theta}$, where $\bGamma\in\R^{p\times p}$ is positive definite. 

The overall objective of this adaptive control architecture is to ensure that the resulting system trajectory satisfies $\bx(t)\in\mathcal{C}$ for all $t$, where $\mathcal{C}\subset\R^n$ is a candidate safe set as defined in \eqref{eq:C}. This objective can be accomplished when the function $h$ defining $\mathcal{C}$ can be certified as an aCBF.

\begin{definition}[\cite{AndrewACC20}]\label{def:acbf}
A continuously differentiable function $h\,:\,\R^n\rightarrow\R$ defining a set $\mathcal{C}$ as in \eqref{eq:C} is said to be an aCBF for \eqref{eq:general_nonlinear_system} if for each $(\bx,\bhat{\theta})\in\R^n\times\R^p$:
\begin{equation}\label{eq:acbf}
    \sup_{\bu\in\R^m}\left\{\tpdv{h}{\bx}(\bx)\left(\bf(\bx) + \bG(\bx)\left[\bu + \bPhi(\bx)\bhat{\theta} \right] \right) \right\} > 0.
\end{equation}
\end{definition}
Controllers meeting the above requirement can be synthesized using a QP approach similar to that in \eqref{eq:cbf-qp}, where the constraint, and hence the resulting controller, depends on the parameter estimate $\bhat{\theta}$; cf. \cite{AndrewACC20}. Safety analysis of this controller is facilitated by the augmented barrier candidate:
\begin{equation}\label{eq:ha}
    h_a(\bx,\bhat{\theta}) \coloneqq h(\bx) - \tfrac{1}{2}(
    \btheta - \bhat{\theta})\T\bGamma^{-1}(\btheta - \bhat{\theta}
    ),
\end{equation}
which defines a set $\mathcal{C}_{a}\subset\R^{n+p}$ in the state space of \eqref{eq:augmented_adaptive_ss} as:
\begin{equation}\label{eq:Ca}
    \mathcal{C}_a \coloneqq \{(\bx,\bhat{\theta})\in\R^{n+p}\mid h_a(\bx,\bhat{\theta}) \geq 0\}.
\end{equation}
Note that $h_a(\bx,\bhat{\theta})\geq0$ implies that $h(\bx)\geq0$, so establishing forward invariance of $\mathcal{C}_a$ is sufficient to accomplish the control objective of ensuring that $\bx(t)\in\mathcal{C}$. The following theorem outlines the main  result with regard to aCBFs.

\begin{theorem}[\cite{AndrewACC20}]\label{thm:acbf}
    Let $h\,:\,\R^n\rightarrow\R$ be an aCBF for \eqref{eq:general_nonlinear_system}, let $\bk\,:\,\R^{n+p}\rightarrow\R^m$ be any locally Lipschitz controller satisfying \eqref{eq:acbf}, and choose the parameter update law:
    \begin{equation}\label{eq:update-law-acbf}
        \bdothat{\theta} = -\bGamma\bpsi(\bx) =  -\bGamma\left(\tpdv{h}{\bx}(\bx)\bG(\bx)\bPhi(\bx) \right)\T,
    \end{equation}
    which, together, define the adaptive closed-loop system \eqref{eq:augmented_adaptive_ss}. Then, the set $\mathcal{C}_a$ as in \eqref{eq:Ca} is forward invariant for \eqref{eq:augmented_adaptive_ss}.
\end{theorem}
There are two notable distinctions between the invariance results when using CBFs and aCBFs. First, the set $\mathcal{C}$ is \emph{not} forward invariant; rather, the augmented set $\mathcal{C}_{a}$ is forward invariant. Thus, to guarantee that $\bx(t)\in\mathcal{C}$ for all time, \eqref{eq:augmented_adaptive_ss} must start in $\mathcal{C}_{a}$, a sufficient condition for which is:
\begin{equation}\label{eq:acbf-gain-condition}
    h(\bx_0) \geq (\|\btheta - \bhat{\theta}_0\|^2)/(2\lambda_{\min}(\bGamma)),
\end{equation}
where $(\bx_0,\bhat{\theta}_0)\in\R^{n+p}$ is the initial condition of the adaptive closed-loop system \eqref{eq:augmented_adaptive_ss} and $\lambda_{\min}(\cdot)$ returns the minimum eigenvalue. This condition can be satisfied by i) initializing \eqref{eq:general_nonlinear_system} with a large safety margin $h(\bx_0)>0$, ii) initializing $\bhat{\theta}_0$ close to $\btheta$, iii) choosing large adaptation gains $\bGamma$, or some combination thereof. Second, the right-hand-side of the aCBF condition \eqref{eq:acbf} is stricter than that in \eqref{eq:cbf}, requiring the existence of inputs that instantaneously increase the value of $h$. In contrast, the $\alpha\in\mathcal{K}_{\infty}^e$ in \eqref{eq:cbf} allows $h$ to decrease, so long as it does not become negative. This second distinction can be rectified through the use of \emph{robust} aCBFs \cite{LopezLCSS21}.

\smallskip
\noindent\textbf{Robust Adaptive Safety.}
\label{sec:robust_adaptive_safety}
We now review RaCBFs \cite{LopezLCSS21}, which reduce conservatism of aCBFs by allowing one to place $\alpha\in\mathcal{K}_{\infty}^e$ on the right-hand-side of the aCBF condition \eqref{eq:acbf}. This approach is facilitated by assuming knowledge of bounds on the parameter estimation error in the sense that for any $\bhat{\theta}$ there exists another vector $\btilde{\vartheta}\in\R^p$ such that:
\begin{equation}\label{eq:parameter-error-bound}
    \|\btheta - \bhat{\theta}\| \leq \|\btilde{\vartheta}\|.
\end{equation}
This bound can be formed using, e.g., parameter projection \cite[App. E]{Krstic}, set membership identification \cite{LopezLCSS21}, or concurrent learning \cite{IsalyACC21,Cohen}, and leads to the notion of a RaCBF.

\begin{definition}[\cite{LopezLCSS21}]\label{def:racbf}
A continuously differentiable function $h\,:\,\R^n\rightarrow\R$ defining a set $\mathcal{C}$ as in \eqref{eq:C} is said to be a RaCBF for \eqref{eq:general_nonlinear_system} if there exists $\alpha\in\mathcal{K}_{\infty}^e$ such that:
\begin{equation}\label{eq:racbf}
    \begin{aligned}
        \sup_{\bu\in\R^m} \dot{h}(\bx,\bhat{\theta},\bu) > - \alpha\big(h(\bx) - \tfrac{1}{2}\btilde{\vartheta}\T\bGamma^{-1}\btilde{\vartheta} \big),
    \end{aligned}
\end{equation}
for each $(\bx,\bhat{\theta})\in\R^n\times\R^p$, 
where:
\begin{equation}\label{eq:h-dot-racbf}
    \dot{h}(\bx,\bhat{\theta},\bu) = \tpdv{h}{\bx}(\bx)\left(\bf(\bx) + \bG(\bx)\left[\bu + \bPhi(\bx)\bhat{\theta} \right] \right),
\end{equation}
and $\btilde{\vartheta}\in\R^p$ is any vector satisfying \eqref{eq:parameter-error-bound}.
\end{definition}
Safety properties of RaCBFs are established using the augmented barrier from \eqref{eq:ha} and are outlined below.

\begin{theorem}[\cite{LopezLCSS21}]\label{thm:racbf}
    Let $h\,:\,\R^n\rightarrow\R$ be a RaCBF for \eqref{eq:general_nonlinear_system}, let $\bk\,:\,\R^{n+p}\rightarrow\R^m$ be any locally Lipschitz controller satisfying \eqref{eq:racbf}, and choose the parameter update law as in \eqref{eq:update-law-acbf},
    which, together, define the adaptive closed-loop system \eqref{eq:augmented_adaptive_ss}. Then, the set $\mathcal{C}_a$ as in \eqref{eq:Ca} is forward invariant for \eqref{eq:augmented_adaptive_ss}.
\end{theorem}

The conclusion of Theorem \ref{thm:racbf} is essentially the same as that of Theorem \ref{thm:acbf}. The key difference lies in the transient behavior produced by aCBFs \eqref{eq:acbf} vs RaCBFs \eqref{eq:racbf}. The former ensures that \eqref{eq:ha} satisfies $\dot{h}_{a}\geq0$ while the latter ensures that $\dot{h}_{a} \geq - \alpha(h_{a})$, a less restrictive condition that only renders the zero superlevel set of $h_{a}$ forward invariant (cf. \cite{AmesTAC17} for a discussion on the difference between these two conditions). To ensure that $(\bx_0,\bhat{\theta}_0)\in\mathcal{C}_{a}$, one must still ensure that \eqref{eq:acbf-gain-condition} holds. In practice, where one may not be able to arbitrarily increase $h(\bx_0)$ or decrease $\|\btheta - \bhat{\theta}_0\|$, the only choice one may have to satisfy \eqref{eq:acbf-gain-condition} is choose a $\bGamma$, a design parameter, with a large minimum eigenvalue. Historically, large adaptation gains are often discouraged as they can excite unmodeled dynamics, amplify measurement noise, saturate actuators, and induce undesirable and unstable oscillatory behavior. In this paper, we present an approach to adaptive safety that leverages high-order tuners \cite{hotadaptive,DanielLeACC22,TeelCDC22}, providing more flexibility for satisfying \eqref{eq:acbf-gain-condition} and taking steps toward mitigating these limitations. 

\begin{remark}
    Our preliminary definitions and theorems on aCBFs and RaCBFs are specializations of those from \cite{AndrewACC20,LopezLCSS21} to the case of uncertain systems with matched parametric uncertainties \eqref{eq:general_nonlinear_system}, where aCBFs and RaCBFs can be made independent of the parameter estimates, cf. \cite[Rem. 2]{AndrewACC20}. 
\end{remark}

\section{Safe High Order Tuners}
\label{sec:hot}
We now present an approach to adaptive safety that addresses the limitations outlined in the previous section. Our approach leverages \emph{high order tuners}, 
which, as we demonstrate herein,
helps decouple the adaptation gains from conditions needed to start in the augmented safe set \eqref{eq:acbf-gain-condition}. 
In the context of adaptive safety, this high order tuner is:
\begin{equation}\label{eq:update-law-hot}
        \dot{\bm\nu}=  -\bGamma\bpsi(\bx),\qquad
        \bdothat{\theta} = \beta \bm \Gamma(\bm\nu- \hat \btheta), 
\end{equation}
where $h\,:\,\R^n\rightarrow\R$ is from \eqref{eq:C}, $\bpsi$ is from \eqref{eq:update-law-acbf}, $\bnu\in\R^p$ is an intermediate parameter estimate, $\bhat{\theta}$ is the parameter estimate used by the downstream controller $\bu=\bk(\bx,\bhat{\theta})$, $\bGamma\in\R^{p\times p}$ is a positive definite adaptation gain matrix, and $\beta>0$ is another adaptation gain. While various insights into high order tuners are established in \cite{hotadaptive}, we note here that \eqref{eq:update-law-hot} is simply a filtered version of the aCBF/RaCBF update law from \eqref{eq:update-law-acbf}. The intermediate estimate $\bnu$ is generated using the same dynamics as in \eqref{eq:update-law-acbf}, which is then sent through a low-pass filter  with bandwidth proportional to $\beta$ to produce $\bhat{\theta}$, the parameter estimate used by the adaptive controller.
Properties of this approach established by the studying:
\begin{align}
    \begin{bmatrix} 
        \dot{\bx} \\ \dot{\bnu} \\ \bdothat{\theta} 
    \end{bmatrix} 
    = 
    \begin{bmatrix}
        \bf(\bx) + \bG(\bx)\big(\bk(\bx,\bhat{\theta}) + \bPhi(\bx)\btheta \big) \\
        -\bGamma\bpsi(\bx) \\ 
        \beta\bGamma(\bnu - \bhat{\theta})
    \end{bmatrix},
\label{eq:augmented_adaptive_hot_ss}
\end{align}
with state $(\bx,\bnu,\bhat{\theta})\in\R^{n + 2p}$. 
The design of the adaptive controller in \eqref{eq:augmented_adaptive_hot_ss} is similar to that of RaCBFs by assuming that for any $\bnu$ there exists another vector $\btilde{\vartheta}\in\R^p$ bounding the intermediate estimation error:
\begin{equation}\label{eq:parameter-error-bound-nu}
    \|\btheta - \bnu\| \leq \|\btilde{\vartheta}\|.
\end{equation}
Control 
is then facilitated by the notion of a \emph{tunable} RaCBF.

\begin{definition}\label{def:tunable_racbf}
A continuously differentiable function $h\,:\,\R^n\rightarrow\R$ defining a set $\mathcal{C}$ as in \eqref{eq:C} is said to be a Tunable RaCBF (T-RaCBF) for \eqref{eq:general_nonlinear_system} if there exists $\alpha>0$ such that:
\begin{equation}\label{eq:tunable_racbf}
    \begin{aligned}
        \sup_{\bu\in\R^m} \dot{h}(\bx,\bhat{\theta},\bu) > - \alpha\left[h(\bx) - \tfrac{1}{2}\btilde{\vartheta}\T\bGamma^{-1}\btilde{\vartheta} \right]+ \tfrac{2}{\beta}\|\bpsi(\bx)\|^2,
    \end{aligned}
\end{equation}
for each $(\bx,\bhat{\theta})\in\R^n\times\R^p$, where $\dot{h}$ is defined in \eqref{eq:h-dot-racbf}, and $\btilde{\vartheta}\in\R^p$ is any vector satisfying \eqref{eq:parameter-error-bound-nu}.
\end{definition}

The difference between RaCBFs and T-RaCBFs is the term $\tfrac{2}{\beta}\|\bpsi(\bx)\|^2$, an input-to-state safe \cite{AmesLCSS19} inspired robustness margin, not present in Def. \ref{def:racbf}, that compensates for transient behavior of the high order tuner \eqref{eq:update-law-hot}. In the limit as $\beta\rightarrow\infty$ this term vanishes and creates a low-pass filter with infinite bandwidth so that $\bhat{\theta}\approx\bnu$, reducing the  T-RaCBF condition \eqref{eq:tunable_racbf} and update law \eqref{eq:update-law-hot} to those of the RaCBF. As $\beta\rightarrow0$ the effective adaptation gain $\beta\bGamma$ in \eqref{eq:update-law-hot} also approaches zero; this lack of adaptation is compensated for with a larger robustness margin in the controller \eqref{eq:tunable_racbf}. The flexibility offered by $\beta$, balancing adaptation and robustness, is the key to mitigating the challenges in satisfying \eqref{eq:acbf-gain-condition}.

Similar to the previous section, safety analysis of this strategy is facilitated by an augmented barrier candidate: 
\begin{align}\label{eq:ha-tunable}
    h_a(\bx,\bnu,\bhat{\theta}) &\coloneqq h(\bx) - \tfrac{1}{2}\btilde{\theta}\T\bGamma^{-1}\btilde{\theta}
    - \tfrac{1}{2}\btilde{\nu}\T\bGamma^{-1}\btilde{\nu},
\end{align}
where $\btilde{\theta}\coloneqq \btheta - \bnu$ and $\btilde{\nu}\coloneqq \bnu - \bhat{\theta}$,
which defines a set $\mathcal{C}_{a}\subset\R^{n+2p}$ in the state space of \eqref{eq:augmented_adaptive_hot_ss} as:
\begin{equation}\label{eq:Ca-hot}
    \mathcal{C}_a \coloneqq \{(\bx,\bm\nu,\bhat{\theta})\in\R^{n+2p}\mid h_a(\bx,\bm\nu,\bhat{\theta}) \geq 0\}.
\end{equation}
Further, $h_a(\bx,\bm\nu,\bhat{\theta}) \ge 0$ implies that $h(\bx)\geq0$ so that establishing forward invariance of $\mathcal{C}_a$ is sufficient to accomplish the control objective of ensuring that $\bx(t)\in\mathcal{C}$. The following theorem outlines the main result with regard to T-RaCBFs.

\begin{theorem}\label{thm:t-racbf}
    Let $h\,:\,\R^n\rightarrow\R$ be a T-RaCBF for \eqref{eq:general_nonlinear_system}, let $\bk\,:\,\R^{n+p}\rightarrow\R^m$ be any locally Lipschitz controller satisfying \eqref{eq:tunable_racbf}, and choose the update laws as in \eqref{eq:update-law-hot},
    which form the adaptive closed-loop system \eqref{eq:augmented_adaptive_hot_ss}. Provided that:
    \begin{equation}\label{eq:beta-condition}
        \beta \geq \alpha/\lambda_{\min}(\bGamma),
    \end{equation}
    then,
    the set $\mathcal{C}_a$ as in \eqref{eq:Ca-hot} is forward invariant for \eqref{eq:augmented_adaptive_hot_ss}.
\end{theorem}

\begin{proof}
Differentiating \eqref{eq:ha-tunable} along trajectories of \eqref{eq:augmented_adaptive_hot_ss}
yields:
\begin{align}
\dot h_{a}(\bx,\bnu,\bhat{\theta})
&= \tpdv{h}{\bx}(\bx)
\big[\bf(\bx)+\bG(\bx)(\bu+\bm{\Phi}(\bx)\bm{\theta})\big]
\\ \nonumber &+ \tilde{\bm{\theta}}^T \bm{\Gamma}^{-1}\dot{\bm{\nu}}
- \tilde{\bm{\nu}}^T \bm{\Gamma}^{-1}(\dot{\bm{\nu}}-\bdothat{\theta}),
\end{align}
where $\bu=\bk(\bx,\bhat{\theta})$. Adding a couple of zeros then yields:
\begin{equation}
    \begin{aligned}
        \dot h_{a} & = \tpdv{h}{\bx}(\bx)[\bf(\bx)+\bG(\bx)\bu]  
        + \bpsi(\bx)\T\bm{\theta}
        + \tilde{\bm{\theta}}^T \bm{\Gamma}^{-1}\dot{\bm{\nu}}
        \\  &- \tilde{\bm{\nu}}^T \bm{\Gamma}^{-1}(\dot{\bm{\nu}}-\bdothat{\theta}) 
        \pm \bpsi(\bx)\T\bm{\nu}
        \pm \bpsi(\bx)\T\hat{\bm{\theta}},
    \end{aligned}
\end{equation}
where functional dependencies on $\dot{h}_a$ have been suppressed for brevity. 
Rearranging terms yields:
\begin{equation*}
    \begin{aligned}
        \dot h_{a}
        &= \tpdv{h}{\bx}(\bx)[\bf(\bx)+\bG(\bx)\bu]
        + \bpsi(\bx)\T(\bm{\theta}-\bm{\nu})
        \\ &+ \tilde{\bm{\theta}}^T \bm{\Gamma}^{-1}\dot{\bm{\nu}} 
        - \tilde{\bm{\nu}}^T \bm{\Gamma}^{-1}(\dot{\bm{\nu}}-\bdothat{\theta})
         + \bpsi(\bx)\T\bm{\nu}
         \pm \bpsi(\bx)\T\hat{\bm{\theta}}.
    \end{aligned}
\end{equation*}
Using the definitions of $\btilde{\theta}$, $\btilde{\nu}$ then yields:
\begin{equation}
    \begin{aligned}
        \dot h_{a}
        &= \tpdv{h}{\bx}(\bx)[\bf(\bx)+\bG(\bx)\bu]
        + \tilde{\bm{\theta}}\T
        \left(\bm{\Gamma}^{-1}\dot{\bm{\nu}}
        + \bpsi(\bx)\right)\\  &
        - \tilde{\bm{\nu}}^T \bm{\Gamma}^{-1}(\dot{\bm{\nu}}-\bdothat{\theta})
        + \bpsi(\bx)\T(\bm{\nu}-\hat{\bm{\theta}})
        + \bpsi(\bx)\T\hat{\bm{\theta}},
    \end{aligned}
\end{equation}
which further simplifies to:
\begin{equation*}
    \begin{aligned}
        \dot h_{a}
        &= \tpdv{h}{\bx}(\bx)
        [\bf(\bx)+\bG(\bx)(\bu+\bm{\Phi}(\bx)\hat{\bm{\theta}})] \\  &
        + \tilde{\bm{\theta}}\T
        \left(\bm{\Gamma}^{-1}\dot{\bm{\nu}}
        + \bpsi(\bx)\right)
        - \tilde{\bm{\nu}}\T
        \left(
        \bm{\Gamma}^{-1}(\dot{\bm{\nu}}-\bdothat{\theta})
        - \bpsi(\bx)
        \right).
    \end{aligned}
\end{equation*}
Substituting the update laws \eqref{eq:update-law-hot} into the above yields:
\begin{equation}
    \begin{aligned}
        \dot h_{a}
        &= \tpdv{h}{\bx}(\bx)
        [\bf(\bx)+\bG(\bx)(\bu+\bm{\Phi}(\bx)\hat{\bm{\theta}})]
        \\  &+ \beta \|\btilde{\nu}\|^2
        + 2\tilde{\bm{\nu}}\T
        \bpsi(\bx).
    \end{aligned}
\end{equation}
Using the assumption that $\bu=\bk(\bx,\bhat{\theta})$ satisfies \eqref{eq:tunable_racbf}
allows $\dot{h}_a$ to be lower bounded as:
\begin{equation}
\begin{aligned}
    \dot h_{a} &\ge 
-\alpha[h(\bx)-\tfrac{1}{2}\tilde{\bm\vartheta}\T\bm\Gamma^{-1}\tilde{\bm\vartheta}] \\ & + \tfrac{2}{\beta}\|\bpsi(\bx)\|^2 + \beta \|\btilde{\nu}\|^2
+ 2\tilde{\bm{\nu}}\T
\bpsi(\bx).
\end{aligned}
\end{equation}
Noting that:
\begin{equation*}
    \tfrac{\beta}{2}\|\btilde{\nu}\|^2 + \tfrac{2}{\beta}\|\bpsi(\bx)\|^2 + 2\btilde{\nu}\T\bpsi(\bx) = \tfrac{\beta}{2}\|\btilde{\nu} + \tfrac{2}{\beta}\bpsi(\bx)\|^2 \geq 0,
\end{equation*}
allows $\dot{h}_a$ to be further lower bounded as:
\begin{equation*}
    \dot{h}_a \geq -\alpha[h(\bx)-\tfrac{1}{2}\btilde{\vartheta}\T\bm\Gamma^{-1}\btilde{\vartheta}] + \tfrac{\beta}{2}\|\btilde{\nu}\|^2.
\end{equation*}
The definition of $h_a$ in \eqref{eq:ha-tunable} and the bound in \eqref{eq:parameter-error-bound-nu} implies:
\begin{align}
    h_a \ge h(\bx)  - \tfrac{1}{2}\btilde{\vartheta}\T\bGamma^{-1}\btilde{\vartheta}
    - \tfrac{1}{2}\btilde{\nu}\T\bGamma^{-1}\btilde{\nu},
\end{align}
from which it follows that:
\begin{equation*}
    \begin{aligned}
        \dot{h}_a \geq & -\alpha[h_a+\tfrac{1}{2}\btilde{\nu}\T\bm\Gamma^{-1}\btilde{\nu}] + \frac{\beta}{2}\|\btilde{\nu}\|^2 \\
        \geq & -\alpha h_{a} + \tfrac{1}{2}[\beta - \tfrac{\alpha}{\lambda_{\min}(\bGamma)}]\|\btilde{\nu}\|^2.
    \end{aligned}
\end{equation*}
Provided \eqref{eq:beta-condition} holds, $\dot{h}_a$ may be bounded as
$\dot{h}_{a}(\bx, \bnu,\bhat{\theta}) \ge -\alpha h_{a}(\bx, \bnu,\bhat{\theta})$
for all $(\bx,\bnu,\bhat{\theta})\in\R^{n + 2p}$, which implies that $h_{a}$ is a barrier function for \eqref{eq:augmented_adaptive_hot_ss}.
Invoking Theorem \ref{thm:barrier} then implies that $\mathcal{C}_{a}$ from \eqref{eq:Ca-hot} is forward invariant for \eqref{eq:augmented_adaptive_hot_ss}.
\end{proof}

Theorem \ref{thm:t-racbf}
establishes forward invariance of a safe set in the joint state space of plant states and parameter estimates. Like Theorem \ref{thm:racbf}, it does not just establish invariance but shows that \eqref{eq:ha-tunable} is a barrier function in the sense of Def. \ref{def:barrier}. To use Theorem \ref{thm:t-racbf} to establish that $\bx(t)\in\mathcal{C}$, one must ensure that the initial conditions of \eqref{eq:augmented_adaptive_hot_ss} satisfy $h_a(\bx_0,\bnu_0,\bhat{\theta}_0)\geq0$. By choosing $\bnu_0=\bhat{\theta}_0$ this can be accomplished provided:
\begin{equation}\label{eq:t-racbf-gain-condition}
    h(\bx_0) \geq \frac{\|\btheta - \bnu_0\|^2}{2\lambda_{\min}(\bGamma)} \geq \frac{\|\btilde{\vartheta}_0\|^2}{2\lambda_{\min}(\bGamma)},
\end{equation}
which is similar to \eqref{eq:acbf-gain-condition}. Despite these similarities, the benefit of the proposed approach lies in the role that $\bGamma$ plays in parameter estimation.
Higher values of $\bGamma$ needed to satisfy \eqref{eq:t-racbf-gain-condition} generate the \emph{intermediate} estimate $\bnu$ in \eqref{eq:update-law-hot}, which is not directly used in \eqref{eq:tunable_racbf} to compute inputs. These inputs instead leverage $\bhat{\theta}$, a filtered version of $\bnu$ generated in \eqref{eq:update-law-hot}, with an adaptation gain dictated by $\beta$. Thus, larger $\bGamma$ may be offset by decreasing $\beta$, provided \eqref{eq:beta-condition} holds. The practical benefits of this approach will be illustrated in Sec. \ref{sec:sim}.



\section{Adaptive Safety for Robotic Systems}



\label{sec:robotic}
In this section, we extend our framework to fully actuated robotic systems modeled using the \emph{manipulator equations}:
\begin{equation}\label{eq:robot-dyn}
    \bM(\bq)\ddot{\bq} + \bC(\bq,\dot{\bq})\dot{\bq} + \bg(\bq) = \bu,
\end{equation}
where $\bq\in\R^n$ is the configuration and $\bu\in\R^n$ is the control input. The term $\bM:\R^n\rightarrow\R^{n\times n}$ is the positive definite inertia matrix, $\bC:\R^{2n}\rightarrow\R^{n\times n}$ captures Coriolis terms, and $\bg:\R^n\rightarrow\R^n$ represents gravitational and other potential forces. A useful property of \eqref{eq:robot-dyn} for adaptive control is that the left-hand-side is linear in the unknown parameters \cite{parikh2019integral}.

\begin{property}\label{property:linear-in-parameters}
    The dynamics in \eqref{eq:robot-dyn} are linear in the unknown parameters $\btheta\in\R^p$ in that there exists a locally Lipschitz $\bY:\R^{4n}\rightarrow\R^{n\times p}$ such that for all $(\bq,\dot{\bq},\br,\dot{\br})\in\R^{4n}$:
    \begin{equation}\label{eq:LIP}
        \bM(\bq)\dot{\br} + \bC(\bq,\dot{\bq})\br + \bg(\bq) = \bY(\bq,\dot{\bq},\br,\dot{\br})\btheta.
    \end{equation}
\end{property}
Although it is possible to invert 
the inertia matrix 
to represent \eqref{eq:robot-dyn} in control affine form \eqref{eq:general_knowen_nonlinear_system}, doing so destroys the structure in \eqref{eq:LIP}. Among other issues, if any such parameters are contained in $\bM$,
then inverting $\bM$ would imply that the control directions $\bG$ are unknown, invalidating the assumptions on the parameters from \eqref{eq:general_nonlinear_system}. Hence, directly working with \eqref{eq:robot-dyn} rather than control affine representations offers various advantages in the context of adaptive control.
Other useful properties of \eqref{eq:robot-dyn} 
are outlined below \cite{parikh2019integral}.
\begin{property}\label{property:skew_semetric}
The inertia $\bM$ and Coriolis matrices $\bC$ satisfy the skew-symmetric property for all $(\bs,\bq,\dot{\bq})\in\R^{3n}$:
\begin{align}
    \bs\T \big[\tfrac{1}{2}\dot \bM(\bq,\dot{\bq})-\bC(\bq,\dot{\bq})\big]\bs = 0.
\end{align}

\end{property}

\begin{property}\label{property:inertia-bound}
    There exist constants $\underline{M},\overline{M}>0$ such that:
    \begin{equation}
        \underline{M}\|\bs\|^2 \le \bs^\top \bM(\bq)\bs \le \overline{M}\|\bs\|^2,\quad\forall(\bq,\bs)\in\R^{2n}.
    \end{equation}
\end{property}

Our main objective is to develop controllers for \eqref{eq:robot-dyn} that account for parametric uncertainty while guaranteeing that the configuration remains within a prescribed safe set: 
\begin{equation}\label{eq:C-robot}
    \mathcal{C} = \{\bq\in\R^n \mid h(\bq) \ge 0\},
\end{equation}
where $h:\R^n\rightarrow\R$ is continuously differentiable. Our objective can thus be formalized as ensuring that $h(\bq(t))\geq0$ along closed-loop trajectories, i.e., $\bq(t)\in\mathcal{C}$ for $t\ge0$.

We take a similar approach to 
\cite{TamasRAL22,CohenARC24} by using CBFs to generate a safe reference velocity, which is tracked by \eqref{eq:robot-dyn} to ensure safety. This approach can also be seen as leveraging
\emph{sliding variables} $\bs$, a classic approach to adaptive control pioneered by Slotine and Li \cite{SlotineIJRR87,SlotineLi}, wherein driving $\bs\rightarrow\bzero$ implies the satisfaction of desirable properties of \eqref{eq:robot-dyn}, traditionally stabilization. Here, we demonstrate that this approach can also be used for safety. 
To this end, define: 
\begin{equation}\label{eq:sliding-variable}
    \bs = \bs(\bq,\dot{\bq}, t) \coloneqq \dot{\bq} - \br(\bq, t),
\end{equation}
which quantifies the error between $\dot{\bq}$ and a continuously differentiable reference velocity $\br\,:\,\R^n\times\R\rightarrow\R^n$ whose properties will be specified shortly. This variable has relative degree one and is defined such that $\bs\approx\bzero$ implies that $\dot{\bq}\approx\br(\bq,t)$ so that the configuration velocity inherits the properties of $\br$. 
To 
accomplish this,
we leverage a modified version of the canonical Slotine-Li adaptive controller \cite{SlotineIJRR87,SlotineLi} to ensure compatibility with our high order tuner:
\begin{equation}\label{eq:slotine-li-controller}
    \bu = -\bK\bs + \bW(\bq,\dot{\bq},t)\bhat{\theta}  - \frac{2}{\beta}\bW(\bq,\dot{\bq},t)\bW(\bq,\dot{\bq},t)\T\bs,
\end{equation}
\begin{equation}\label{eq:robot_update_1}
    \dot{\bnu} = -\bGamma\bW(\bq,\dot{\bq},t)\T\bs,
\end{equation}
\begin{equation}\label{eq:robot_update_2}
    \bdothat{\theta} = \beta\bGamma(\bnu - \bhat{\theta}),
\end{equation}
where $\bK\in\R^{n\times n}$ is a positive definite gain, $\bGamma\in\R^{p\times p}$, $\beta>0$ are as in \eqref{eq:update-law-hot}, and:
\begin{equation}\label{eq:modified-regressor}
    \bW(\bq,\dot{\bq},t) \coloneqq \bY(\bq,\dot{\bq},\br(\bq,t),\dot{\br}(\bq,\dot{\bq},t)),
\end{equation}
is a modified regressor matrix. The definition in \eqref{eq:modified-regressor} is motivated by the observation that:
\begin{equation}\label{eq:sliding-dynamics}
    \bM(\bq)\dot{\bs} = \bu - \bW(\bq,\dot{\bq},t)\btheta - \bC(\bq,\dot{\bq})\bs,
\end{equation}
which one may verify by differentiating \eqref{eq:sliding-variable}, premultiplying by $\bM$, substituting \eqref{eq:robot-dyn}, and using \eqref{eq:LIP}. 
Performance of this controller is then analyzed using the Lyapunov-like function:
\begin{equation}\label{eq:lyaponov_robot}
    V(\bz,t) = \tfrac{1}{2}\big(
    \bs^{\top}\mathbf{M}(\mathbf{q})\bs
    + \tilde{\boldsymbol{\theta}}^{\top}\boldsymbol{\Gamma}^{-1}\tilde{\boldsymbol{\theta}}
    + \tilde{\boldsymbol{\nu}}^{\top}\boldsymbol{\Gamma}^{-1}\tilde{\boldsymbol{\nu}}
    \big),
\end{equation}
with $\btilde{\theta}\coloneqq\btheta-\bnu$, $\btilde{\nu}\coloneqq\bnu-\bhat{\theta}$, and $\bz\coloneqq(\bq,\dot{\bq},\bnu,\bhat{\theta})\in\R^{2n + 2p}$.
The following lemma establishes a useful property of \eqref{eq:lyaponov_robot}.

\begin{lemma}\label{lemma:asymptotic-tracking}
    Consider system \eqref{eq:robot-dyn} satisfying Properties \ref{property:linear-in-parameters}-\ref{property:inertia-bound} and let $\bs$ be defined as in \eqref{eq:sliding-variable}. 
    The controller \eqref{eq:slotine-li-controller} and update laws \eqref{eq:robot_update_1}-\eqref{eq:robot_update_2} ensure that \eqref{eq:lyaponov_robot} satisfies:
    \begin{equation}\label{eq:vdot_robot}
        \dot{V}(\bz,t) \leq -\lambda_{\min}(\bK)\|\bs\|^2 - \tfrac{\beta}{2}\|\btilde{\nu}\|^2,
    \end{equation}
    for all $(\bz,t)\in\R^{2n + 2p+1}$.
\end{lemma}

\begin{proof}
The time derivative of $V$ from \eqref{eq:lyaponov_robot} is:
\begin{equation*}
    \begin{aligned}
        \dot V
        = & \bs\T \bM(\bq)\dot \bs
        + \tfrac{1}{2}\bs\T \dot{\bM}(\bq,\dot{\bq})\bs
        - \btilde{\theta}\T \bGamma^{-1}\dot{\bnu}
        +\btilde{\nu}^\top \bGamma^{-1}\bdottilde{\nu} \\ 
        = & \bs\T[ \bu - \bW(\bq,\dot{\bq},t)\btheta - \bC(\bq,\dot{\bq})\bs] + \tfrac{1}{2}\bs\T \dot{\bM}(\bq,\dot{\bq})\bs \\ 
        & - \btilde{\theta}\T \bGamma^{-1}\dot{\bnu} +\btilde{\nu}^\top \bGamma^{-1}\bdottilde{\nu} \\ 
        = & \bs\T[ \bu - \bW(\bq,\dot{\bq},t)\btheta]  - \btilde{\theta}\T \bGamma^{-1}\dot{\bnu} +\btilde{\nu}^\top \bGamma^{-1}\bdottilde{\nu},
    \end{aligned}
\end{equation*}
where the second equality follows from \eqref{eq:sliding-dynamics} and the third from Property \ref{property:skew_semetric}. Choosing the control law as in \eqref{eq:slotine-li-controller}
yields:
\begin{equation}
    \begin{aligned}
        \dot V
        &=
        \bs^\top
        [
        -
        \mathbf K\bs-\tfrac{2}{\beta}\bW\bW^{\top}\bs-\bW(\btheta-\bhat{\theta})]
        \\  & - \btilde{\theta}\T \bGamma^{-1}\dot{\bnu}
        +
        \btilde{\nu}^\top
        \boldsymbol{\Gamma}^{-1}
        \bdottilde{\nu},
    \end{aligned}
\end{equation}
where functional dependencies have been suppressed for brevity.
Expanding the above and adding a zero yields:
\begin{equation}
    \begin{aligned}
        \dot V = &- \bs\T\bK\bs - \tfrac{2}{\beta}\bs\T\bW\bW\T\bs-\bs\T\bW(\btheta-\bhat{\theta}) \\
        & - \btilde{\theta}\T \bGamma^{-1}\dot{\bnu} + \btilde{\nu}^\top\bGamma^{-1}\bdottilde{\nu} \pm \bs\T\bW\bnu \\
        = & - \bs\T\bK\bs - \tfrac{2}{\beta}\|\bW\T\bs\|^2 -\bs^{\top}\bW(\btheta - \bnu) \\
        & - \bs^{\top}\bW(\bnu - \bhat{\theta})
        - \btilde{\theta}\T \bGamma^{-1}\dot{\bnu} +
        \btilde{\nu}^\top \bGamma^{-1}\bdottilde{\nu} \\
        = & - \bs\T\bK\bs - \tfrac{2}{\beta}\|\bW\T\bs\|^2 - \btilde{\theta}\T[\bGamma^{-1}\dot{\bnu} + \bW\T\bs] \\
        & + \btilde{\nu}\T[\bGamma^{-1}\dot{\bnu} - \bGamma^{-1}\bdothat{\theta} - \bW\T\bs],
    \end{aligned}
\end{equation}
where the final equality follows from the definitions of $\btilde{\theta}$ and $\btilde{\nu}$.
Choosing the first adaptation law as \eqref{eq:robot_update_1}
leads to:
\begin{equation*}
    \begin{aligned}
        \dot V = & - \bs\T\bK\bs - \tfrac{2}{\beta}\|\bW\T\bs\|^2 -\btilde{\nu}\T[\bGamma^{-1}\bdothat{\theta} + 2\bW\T\bs],
    \end{aligned}
\end{equation*}
and taking the second adaptation law as in \eqref{eq:robot_update_2} yields:
\begin{equation}
    \begin{aligned}
        \dot V = & - \bs\T\bK\bs - \tfrac{2}{\beta}\|\bW\T\bs\|^2 - \beta\|\btilde{\nu}\|^2 - 2\btilde{\nu}\T\bW\bs.
    \end{aligned}
\end{equation}
Noting that:
\begin{equation*}
    \tfrac{\beta}{2}\|\btilde{\nu}\|^2 + \tfrac{2}{\beta}\|\bW\T\bs\|^2 + 2\btilde{\nu}\T\bW\T\bs = \tfrac{\beta}{2}\|\btilde{\nu} + \tfrac{2}{\beta}\bW\T\bs\|^2 \geq 0,
\end{equation*}
allows $\dot{V}$ to be upper bounded as:
\begin{equation}
    \dot V \leq - \lambda_{\min}(\bK)\|\bs\|^2 - \tfrac{\beta}{2}\|\btilde{\nu}\|^2 \leq 0,
\end{equation}
demonstrating \eqref{eq:vdot_robot}, as desired.
\end{proof}

Lemma \ref{lemma:asymptotic-tracking} makes no claims of convergence of any signals. As $V$ is bounded from below and $\dot{V}\leq0$ by \eqref{eq:vdot_robot}, along closed-loop trajectories $t\mapsto \bz(t)$, 
$V$ is bounded implying that $\bs(\cdot),\bnu(\cdot),\bhat{\theta}(\cdot)$ are also bounded.
Although 
$\bs$ is bounded,
nothing can be said about boundedness of $\bq(\cdot),\dot{\bq}(\cdot)$ without further assumptions on $\br$. If one follows \cite{SlotineIJRR87} and takes:
\begin{equation}\label{eq:rd}
    \br(\bq,t)  = \br_{\rm{d}}(\bq,t) \coloneqq \dot{\bq}_{d}(t) + \bLambda(\bq - \bq_{\rm{d}}(t)),
\end{equation}
for a desired trajectory $\bq_{\rm{d}}(\cdot)\in\mathcal{L}_{\infty}$ and positive definite $\bLambda\in\R^{n\times n}$, then $\bs(\cdot)\in\mathcal{L}_{\infty}$ would imply that $\bq(\cdot),\dot{\bq}(\cdot)\in\mathcal{L}_{\infty}$, which could then be combined with a signal chasing argument and Barbalat's Lemma \cite[Thm. A.6]{Krstic} to establish boundedness of all closed-loop system signals and that $\bs\rightarrow\bzero$, cf. \cite{DanielLeACC22}. While boundedness of all closed-loop system signals is desirable in practice, it is not mathematically necessary for forward invariance if the safe set is unbounded\footnote{Example: $\mathcal{C}=\{x\in\R\mid x\geq0\}$ is forward invariant for $\dot{x} = x$.}.

\begin{remark}
    A similar control strategy to that above is presented in \cite{DanielLeACC22}.
    A key distinction between our approach and \cite{DanielLeACC22} is that the controller in \cite{DanielLeACC22} directly cancels cross terms involving the intermediate estimate $\bnu$, whereas \eqref{eq:slotine-li-controller} indirectly compensates for such cross terms via the nonlinear damping term $\bW\bW\T\bs$, thereby ensuring that \eqref{eq:slotine-li-controller} does not depend on $\bnu$, which may be generated with a large gain $\bGamma$. 
\end{remark}

Having shown that the adaptive controller \eqref{eq:slotine-li-controller}-\eqref{eq:robot_update_2} ensures $\bs\in\mathcal{L}_{\infty}$, we now outline the properties that $\br$ should satisfy so that boundedness of $\bs$ implies satisfaction of the constraint \eqref{eq:C-robot}. 
Inspired by \eqref{eq:tunable_racbf}, we design $\br$ such that:
\begin{equation}\label{eq:tunable_racbf_robotic}
    \begin{aligned}
        \tpdv{h}{\bq}(\bq) \br \geq - \alpha\left[h(\bq) - \tfrac{1}{2\mu}\btilde{\vartheta}\T\bGamma^{-1}\btilde{\vartheta} \right]+ \tfrac{1}{\epsilon}\|\tfrac{\partial h}{\partial \bq}(\bq)\|^2,
    \end{aligned}
\end{equation}
for each 
$\bq\in\R^n$,
where $\btilde{\vartheta}\in\R^p$  satisfies \eqref{eq:parameter-error-bound-nu}, and $\alpha>0$, $\mu >0$, $\epsilon>0$.
The interpretation of $\br$ is that of a safe reference velocity. If we could simply take $\dot{\bq}=\br(\bq,t)$, the above would be sufficient to establish forward invariance of $\mathcal{C}$ since $\dot{h}=\pdv{h}{\bq}\br$ would satisfy \eqref{eq:tunable_racbf_robotic}.
Note that a differentiable $\br$ satisfying \eqref{eq:tunable_racbf_robotic} exists under relatively mild conditions on $h$ \cite[Lem. 1]{CohenLCSS24} with explicit formulas available in \cite{CohenLCSS23,CohenARC24}. In addition, tracking a desired configuration $\bq_\mathrm{d}(t)$ objective can be incorporated into $\br$ (hence why $\br$ may be time-varying) by 
defining a desired velocity
such as $\br_{\rm{d}}$ in \eqref{eq:rd}
which then would be used as the ``desired controller" in a smooth safety filter \cite{CohenLCSS23,CohenARC24} to produce a differentiable $\br$ satisfying \eqref{eq:tunable_racbf_robotic}.

Since we cannot simply set $\dot{\bq}=\br$, we must choose $\bu$ to maintain $\dot{\bq}$ close to $\br$,
which we accomplish using the adaptive controller in \eqref{eq:slotine-li-controller}.
To formally establish safety properties of the controller from \eqref{eq:slotine-li-controller}, consider the set:
\begin{equation}\label{eq:S-adaptive}
    \mathcal{S}(t) = \{\bz\in\R^{2n+2p} \mid B(\bz, t) \geq 0\},
\end{equation}
\begin{equation}\label{eq:B}
    \begin{aligned}
        B(\bz, t) \coloneqq & h(\bq) - \tfrac{1}{\mu}V(\bz,t), \\ 
    \end{aligned}
\end{equation}
where $h$ is from \eqref{eq:C-robot}, $V$ is the Lyapunov-like function from \eqref{eq:lyaponov_robot}, and $\mu>0$. Note that $B(\bz,t)\geq0 \implies h(\bq)\geq0$ since $V\geq0$. Hence, establishing forward invariance of $\mathcal{S}(t)$ provides a pathway towards guaranteeing satisfaction of the configuration constraint in \eqref{eq:C-robot}. The following result provides conditions under which \eqref{eq:S-adaptive} is forward invariant.

\begin{theorem}\label{th:safety_robot}
     Let the hypotheses of Lemma \ref{lemma:asymptotic-tracking} hold. Let $\br\,:\,\R^n\times\R\rightarrow\R^n$ be any continuously differentiable reference velocity satisfying \eqref{eq:tunable_racbf_robotic}. Provided that \eqref{eq:beta-condition} holds, and:
    \begin{equation}\label{eq:safe-tracking-conditions}
          \lambda_{\min}(\bK) \geq \max\left\{\frac{\epsilon\mu}{2},~\alpha\overline{M}\right\},
    \end{equation}
    then $\mathcal{S}$ as in \eqref{eq:S-adaptive} is forward invariant for \eqref{eq:robot-dyn}, \eqref{eq:robot_update_1}-\eqref{eq:robot_update_2}.
\end{theorem}

\begin{proof}
Taking time derivative of $B(\bz,t)$ yields:
\begin{equation*}
        \dot{B}(\bz,t) =  \tpdv{h}{\bq}(\bq)\dot{\bq} - \tfrac{1}{\mu}\dot{V}(\bz,t) =  \tpdv{h}{\bq}(\bq)\br + \tpdv{h}{\bq}(\bq)\bs - \tfrac{1}{\mu}\dot{V}(\bz,t),
\end{equation*}
where the second equality follows from \eqref{eq:sliding-variable}.
Lower bounding the above using \eqref{eq:vdot_robot} and \eqref{eq:tunable_racbf_robotic} leads to:
\begin{equation}
    \begin{aligned}
        \dot{B} \geq & -\alpha h + \tfrac{\alpha}{2\mu}\btilde{\vartheta}\T\bGamma^{-1}\btilde{\vartheta} + \tfrac{1}{\epsilon}\|\tpdv{h}{\bq}\|^2 - \|\tpdv{h}{\bq}\|\|\bs\| \\ 
        & + \tfrac{\lambda_{\min}(\bK)}{\mu}\|\bs\|^2 + \tfrac{\beta}{2\mu}\|\btilde{\nu}\|^2,
    \end{aligned}
\end{equation}
where we used $\tpdv{h}{\bq}\bs\geq-|\tpdv{h}{\bq}\bs|\geq-\|\tpdv{h}{\bq}\|\|\bs\|$ which follows from the Cauchy-Schwarz inequality. Now, note that:
\begin{equation*}
    \begin{aligned}
        \tfrac{1}{\epsilon}\|\tpdv{h}{\bq}\|^2 - \|\tpdv{h}{\bq}\|\|\bs\| + \tfrac{\lambda_{\min}(\bK)}{2\mu}\|\bs\|^2 \\ 
        = 
        \begin{bmatrix}
          \|\pdv{h}{\bq}\| & \|\bs\| 
        \end{bmatrix} 
        \begin{bmatrix}
            \frac{1}{\epsilon} & -\frac{1}{2}\\ -\frac{1}{2} & \frac{\lambda_\mathrm{min}(\bK)}{2\mu}
        \end{bmatrix}
        \begin{bmatrix}
           \|\pdv{h}{\bq}\| \\ \|\bs\| 
        \end{bmatrix}.
    \end{aligned}
\end{equation*}
Thus, provided that \eqref{eq:safe-tracking-conditions} holds, the above matrix is positive semi-definite, allowing $\dot{B}$ to
be further lower bounded as:
\begin{equation*}
    \dot{B}(\bz,t) \ge -\alpha h + \tfrac{\alpha}{2\mu}\btilde{\vartheta}\T\bGamma^{-1}\btilde{\vartheta} + \tfrac{\lambda_{\min}(\bK)}{2\mu}\|\bs\|^2 + \tfrac{\beta}{2\mu}\|\btilde{\nu}\|^2.
\end{equation*}
Now, using the definition of $B$ from \eqref{eq:B}, we have:
\begin{equation*}
    B(\bz, t) \ge h(\bq)  -\tfrac{\overline{M}}{2\mu}\|\bs\|^2- \tfrac{1}{2\mu}\btilde{\vartheta}\T\bGamma^{-1}\btilde{\vartheta}
    - \tfrac{1}{2\mu}\btilde{\nu}\T\bGamma^{-1}\btilde{\nu},
\end{equation*}
which allows $\dot{B}$ to be further bounded as:
\begin{equation*}
    \begin{aligned}
        \dot{B} \geq & -\alpha[ B + \tfrac{\overline{M}}{2\mu}\|\bs\|^2+\tfrac{1}{2\mu}\btilde{\nu}\T\bm\Gamma^{-1}\btilde{\nu}] \\ 
        & + \tfrac{\lambda_{\min}(\bK)}{2\mu}\|\bs\|^2 + \tfrac{\beta}{2\mu}\|\btilde{\nu}\|^2 \\
        \geq & -\alpha B + \tfrac{1}{2\mu}[\lambda_{\min}(\bK) - \alpha\overline{M}]\|\bs\|^2 \\ 
        &+ \tfrac{1}{2\mu}[\beta - \tfrac{\alpha}{\lambda_{\min}(\bGamma)}]\|\btilde{\nu}\|^2.
    \end{aligned}
\end{equation*}
Provided that \eqref{eq:beta-condition} and  \eqref{eq:safe-tracking-conditions} hold, we may bound $\dot{B}$ as:
\begin{equation}
    \dot{B}(\bz,t) \ge -\alpha(B(\bz, t)),\quad\forall (\bz,t)\in\R^{2n + 2p+1},
\end{equation}
which implies that ${B}(\bz,t)$ is a barrier function\footnote{Safety results for time-varying barrier functions are analagous to their time-invariant counterparts. See \cite[Appendix A]{KondaLCSS21} for further details.} for \eqref{eq:robot-dyn}.
Invoking Theorem \ref{thm:barrier} then implies that the set $\mathcal{S}$ from \eqref{eq:S-adaptive} is forward invariant for the dynamics in \eqref{eq:robot-dyn}, \eqref{eq:robot_update_1}-\eqref{eq:robot_update_2}.
\end{proof}

In addition to the filter bandwidth requirement \eqref{eq:beta-condition}, the main condition of Theorem \ref{th:safety_robot} is \eqref{eq:safe-tracking-conditions}, which requires the control gain $\bK$ in \eqref{eq:slotine-li-controller} to be large enough to compensate for the rate of change of the reference velocity $\br$, dictated by $\alpha$ and $\epsilon$. Similar conditions arise in \cite{TamasRAL22,CohenARC24} on which our strategy is based. Like Theorems \ref{thm:acbf}-\ref{thm:t-racbf}, Theorem \ref{th:safety_robot} does not guarantee invariance of $\mathcal{C}$ from \eqref{eq:C-robot} but invariance of an auxiliary set $\mathcal{S}$ as in \eqref{eq:S-adaptive} whose invariance is sufficient to conclude that $\bq(t)\in\mathcal{C}$. Ensuring that $\bq(t)\in\mathcal{C}$ thus requires starting in $\mathcal{S}$, which, by taking $\bnu_0=\bhat{\theta}_0$ can be achieved via: 
\begin{equation}\label{eq:robot-acbf-gain-condition}
    h(\bx_0) \geq \frac{\overline{M}\|\bs_0\|^2}{2\mu} + \frac{\|\btilde{\vartheta}_0\|^2}{2\mu\lambda_{\min}(\bGamma)},
\end{equation}
with consequences similar to those in \eqref{eq:acbf-gain-condition} and \eqref{eq:t-racbf-gain-condition}.

Finally, if $\mathcal{C}$ is compact, one may establish boundedness of all closed-loop signals and convergence of $\bs$, $\btilde{\nu}$ to zero.
\begin{corollary}
Let the hypotheses of Theorem~\ref{th:safety_robot} hold. Suppose that $\mathcal{C}$ in~\eqref{eq:C-robot} is compact and that $\br(\bq,t)$, $\pdv{\br}{\bq}(\bq,t)$, $\pdv{\br}{t}(\bq,t)$ are uniformly bounded with respect to time. 
Provided $\bz_0\in\mathcal{S}(0)$, then $\bz(\cdot)\in\mathcal{L}_{\infty}$, $\bu(\cdot)\in\mathcal{L}_{\infty}$, and:
\begin{equation}\label{eq:asymptotic-tracking}
\lim_{t\rightarrow\infty}\|\dot{\bq}(t) - \br(\bq(t),t)\| = 0,~
\lim_{t\rightarrow\infty}\|\bnu(t)-\bhat{\theta}(t)\| = 0.
\end{equation}
\end{corollary}

\begin{proof}
    Since $\mathcal{S}$ is forward invariant and $\bz_0\in\mathcal{S}(0)$, we have $\bz(t)\in\mathcal{S}(t)$ for all $t\in I(\bz_0)$, with $I(\bz_0)$ the interval of existence.
    Further, since $\bz(t)\in\mathcal{S}(t)$ implies $\bq(t)\in\mathcal{C}$, and $\mathcal{C}$ is compact, 
    $\bq(t)$ 
    is bounded for all $t\in I(\bz_0)$.
    It follows from Lemma \ref{lemma:asymptotic-tracking} that $\bs(t),\bnu(t),\bhat{\theta}(t)$ are also bounded for $t\in I(\bz_0)$.
    Further, since $\br$ is continuous, 
    and $\br$ is bounded uniformly in $t$, 
    $\br(\bq(t),t)$ is also bounded for $t\in I(\bz_0)$. Using \eqref{eq:sliding-variable}, this also implies that $\dot{\bq}(t)$ is bounded for $t\in I(\bz_0)$. Hence, $\bz(t)$ is bounded for all $t\in I(\bz_0)$, implying $\bz(t)$ lies in a compact set and that $I(\bz_0)=[0,\infty)$ by \cite[Thm. 3.3]{Khalil}, i.e., $\bz\in\mathcal{L}_{\infty}$.
    Now note that $\dot \br(\bq,t) = \tpdv{\br}{\bq}(\bq,t)\dot \bq + \tpdv{\br}{t}(\bq,t)$, which implies that $\dot{\br}\in\mathcal{L}_{\infty}$ since $\br$ is continuously differentiable and the partial derivatives are bounded uniformly in $t$. Since $\bq,\dot{\bq},\br, \dot{\br}\in\mathcal{L}_{\infty}$ it follows that $\bY\in\mathcal{L}_{\infty}$, which implies that $\bW\in\mathcal{L}_{\infty}$. It then follows from \eqref{eq:slotine-li-controller}-\eqref{eq:robot_update_2} that $\bu,\dot{\bnu}, \bdothat{\theta}\in\mathcal{L}_{\infty}$, which also implies that $\bdottilde{\nu}\in\mathcal{L}_{\infty}$. From \eqref{eq:sliding-dynamics} we have $\dot{\bs}\in\mathcal{L}_{\infty}$. Since $\dot{\bs},\bdottilde{\nu}\in\mathcal{L}_{\infty}$, both $\bs$ and $\btilde{\nu}$ are uniformly continuous. Further, it follows \eqref{eq:vdot_robot} that $\bs,\btilde{\nu}\in\mathcal{L}_{2}$. 
    The conclusion in \eqref{eq:asymptotic-tracking}
    then follows from Barbalat's Lemma \cite[Lem. A.6]{Krstic}. 
\end{proof}

\section{Numerical Examples}
\label{sec:sim}
In this section, we illustrate advantages of higher order tuners in adaptive safety and demonstrate how our approach applies to robotic systems \eqref{eq:robot-dyn} via numerical simulations.

\smallskip
\noindent\textbf{Double Integrator.}
First, we compare adaptive safety with T-RaCBFs and RaCBFs for a double integrator modeled as:
\begin{equation}
    \begin{aligned}
    \dot x_1 &= x_2,\\
    \dot x_2 &= \theta_1 x_1 + \theta_2 x_2 + u,
    \end{aligned}
\end{equation}
where $\bx=(x_1,x_2)\in\R^2$ with position $x_1$ and velocity $x_2$, control input $u\in\R$, and unknown parameters $\btheta=(\theta_1,\theta_2)\in\R^2$. 
These dynamics take the form of \eqref{eq:general_nonlinear_system} with $\bf(\bx)=(x_2,0)$, $\bg(\bx)=(0,1)$, and $\bPhi(\bx)=(x_1,x_2)\T$.
The objective is to follow a reference position trajectory $x_{1\mathrm{d}}$ 
while ensuring that the position $x_1$ respects predefined constraints, $\vert x_1(t)\vert \leq x_{1\mathrm{max}}$ for all $t\ge0$. The nominal control law is derived using the backstepping method \cite{Krstic}. The CBF for the position constraint is defined based on CBF backstepping \cite{AndrewCDC22,CohenARC24,CohenLCSS24} and is given by:
\begin{equation}
    h(\bx) = ( x_{1\mathrm{max}}^2 - x_1^2) - \tfrac{1}{\rho}(x_2 + \Delta x_1)^2,
\end{equation}
which defines a set $\mathcal{C}$ as in \eqref{eq:C} and is used to construct controllers meeting the RaCBF \eqref{eq:racbf} and T-RaCBF \eqref{eq:tunable_racbf} conditions using a QP similar to \eqref{eq:cbf-qp}. The RaCBF \eqref{eq:update-law-acbf} and T-RaCBF \eqref{eq:update-law-hot} update laws are used with a projection operator \cite[App. E]{Krstic}
to keep the estimates bounded within a pre-defined set, ensuring the estimation errors also remain bounded within a pre-defined set, as in \eqref{eq:parameter-error-bound} and \eqref{eq:parameter-error-bound-nu}.

Finally, a simulation is run for $10s$, where the reference position is $x_{1\mathrm{d}} = 1.5\sin(2t)$, the unknown parameters are $\theta_1 = 10$, $\theta_2 = 10$, with $\|\btilde{ \vartheta}\| = 14.14$, the position constraint is $ x_{1\mathrm{max}}=1$, and the initial conditions are 
$\bx_0 = (0.75, 0)$, and $\bhat{\theta}_0 = (0, 0)$. Based on 
$\bx_0$, $\bhat{\theta}_0$,
\eqref{eq:acbf-gain-condition} requires $\lambda_{\mathrm{min}}(\bm\Gamma)\ge228$. The safety parameters are set to $\rho = 50$, $\Delta = 0.1$, $\alpha = 2.5$, $\bm\Gamma= 250\bI_2$, and $\beta = 0.05$. The same hyperparameters are chosen for both methods and the only extra tuning knob for the T-RaCBF is the filter bandwidth $\beta$. Fig. \ref{fig:double_integrator_hot_vs_robust} shows the simulation results comparing T-RaCBFs and RaCBFs. 

\begin{figure}[h]
    \centering
    \includegraphics[width=1\columnwidth]{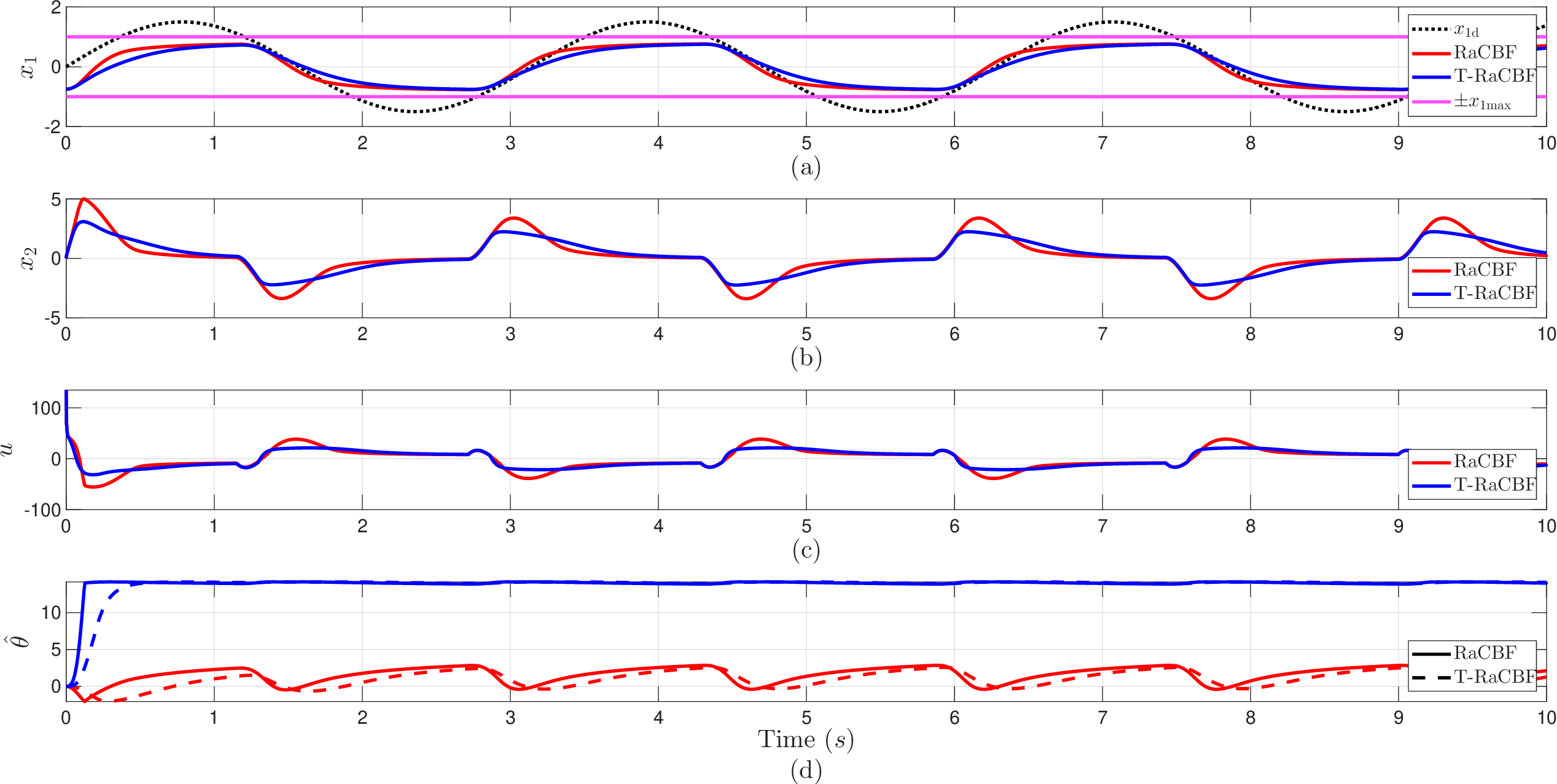}
    \vspace{-8mm}
    \caption{Adaptive safety of double integrator system with T-RaCBF and RaCBF. (a) shows position state $x_1$, (b) shows the velocity state $x_2$, (c) shows control signals, and (d) shows unknown parameter estimates.}
    \label{fig:double_integrator_hot_vs_robust}
    \vspace{-2mm}
\end{figure}

Simulation results demonstrate that both methods ensure safety. However, the T-RaCBF approach yields smoother control inputs and parameter estimates, and requires less control effort, due to the filtering structure of the high order tuner, which allows for modulating the effective adaptation gain via $\beta$ despite requiring large $\bGamma$ to satisfy \eqref{eq:acbf-gain-condition} and \eqref{eq:t-racbf-gain-condition}.

\smallskip
\noindent
\textbf{Robotic Manipulator.}
We use a two-link manipulator \cite{parikh2019integral} as an example of a robotic system \eqref{eq:robot-dyn} with $\bq=(q_1,q_2)\in\R^2$ denoting the angles of the links, and dynamics:
\begin{equation*}
\begin{aligned}
\bM(\bq)&=\begin{bmatrix}
p_1 + 2p_3 c_2 & p_2 + p_3 c_2 \\
p_2 + p_3 c_2 & p_2
\end{bmatrix},
\quad
\bg(\bq)= 
\begin{bmatrix}
    0 \\ 0
\end{bmatrix},
\\
\bC(\bq, \dot \bq) &=
\begin{bmatrix}
- p_3 s_2 \dot{q}_2 & - p_3 s_2 (\dot{q}_1 + \dot{q}_2) \\
p_3 s_2 \dot{q}_1 & 0
\end{bmatrix},
\quad
\bu =
\begin{bmatrix}
\tau_1 \\
\tau_2
\end{bmatrix},
\end{aligned}
\end{equation*}
where $p_1 = 3.473$, $p_2 = 0.196$, and $p_3 = 0.242$, which define $\btheta=(p_1,p_2,p_3)\in\R^3$, with $\|\btilde{ \vartheta}\| = 8.66$, and $\overline{M} = 5$. The terms $c_2$ and $s_2$ are defined as $c_2 \coloneqq  \cos{(q_2)}$, and $s_2 \coloneqq \sin{(q_2)}$. 
Our objective is to track a desired trajectory $\bq_{\rm{d}}(t)$ 
while satisfying constraints on $\bq$ of the form $\vert q_1 \vert \le q_\mathrm{m}$, $\vert q_2 \vert \le q_\mathrm{m}$. The tracking objective is handled with \eqref{eq:rd}. The constraints are handled via the smooth combination \cite{TamasLCSS23}:
\begin{equation}
h(\bq) = -\tfrac{1}{\lambda_h}
\log(
e^{-\lambda_h(q_\mathrm{m}^2 - q_1^2)} + e^{-\lambda_h(q_\mathrm{m}^2 - q_2^2)}
),
\end{equation}
which defines $\mathcal{C}$ as in \eqref{eq:C-robot} and is used in a smooth safety filter \cite{CohenLCSS23,CohenARC24} with $\br_{\rm{d}}$ acting as the desired controller to produce a differentiable $\br$ satisfying \eqref{eq:tunable_racbf_robotic}.
Note that existing approaches to adaptive safety \cite{AndrewACC20,LopezLCSS21,L1CBF,IsalyACC21,Cohen} that consider control affine systems \eqref{eq:general_nonlinear_system} are not compatible with the dynamics considered here, as inverting $\bM$ would destroy the LIP property and induce uncertainty in the control directions.

Simulations are run using the controller \eqref{eq:slotine-li-controller}-\eqref{eq:robot_update_2}, where the desired trajectory is
$\bq_\mathrm{d}(t) = (\frac{\pi}{4}\sin{(2t)}, \frac{\pi}{4}\sin{(2t)})\T$, with an initial condition $\bq_0 = \dot \bq_0 = \bzero$. The control gains are $\bLambda = 0.25\bI_2$, $\bK = 50 \bI_2$, and adaptation gains $\bm\Gamma = 150 \bI_3$, and $\beta = 0.25$. The safety parameters are set to $q_{\mathrm{m}} = \frac{\pi}{6}$, and $\lambda_h = \alpha = \mu = \epsilon = 10$. Fig. \ref{fig:two_link_comparison} shows the simulation results of safe control of the robotic manipulator system. 

\begin{figure}[h]
    \centering
    \vspace{-1mm}
    \includegraphics[width=1\columnwidth]{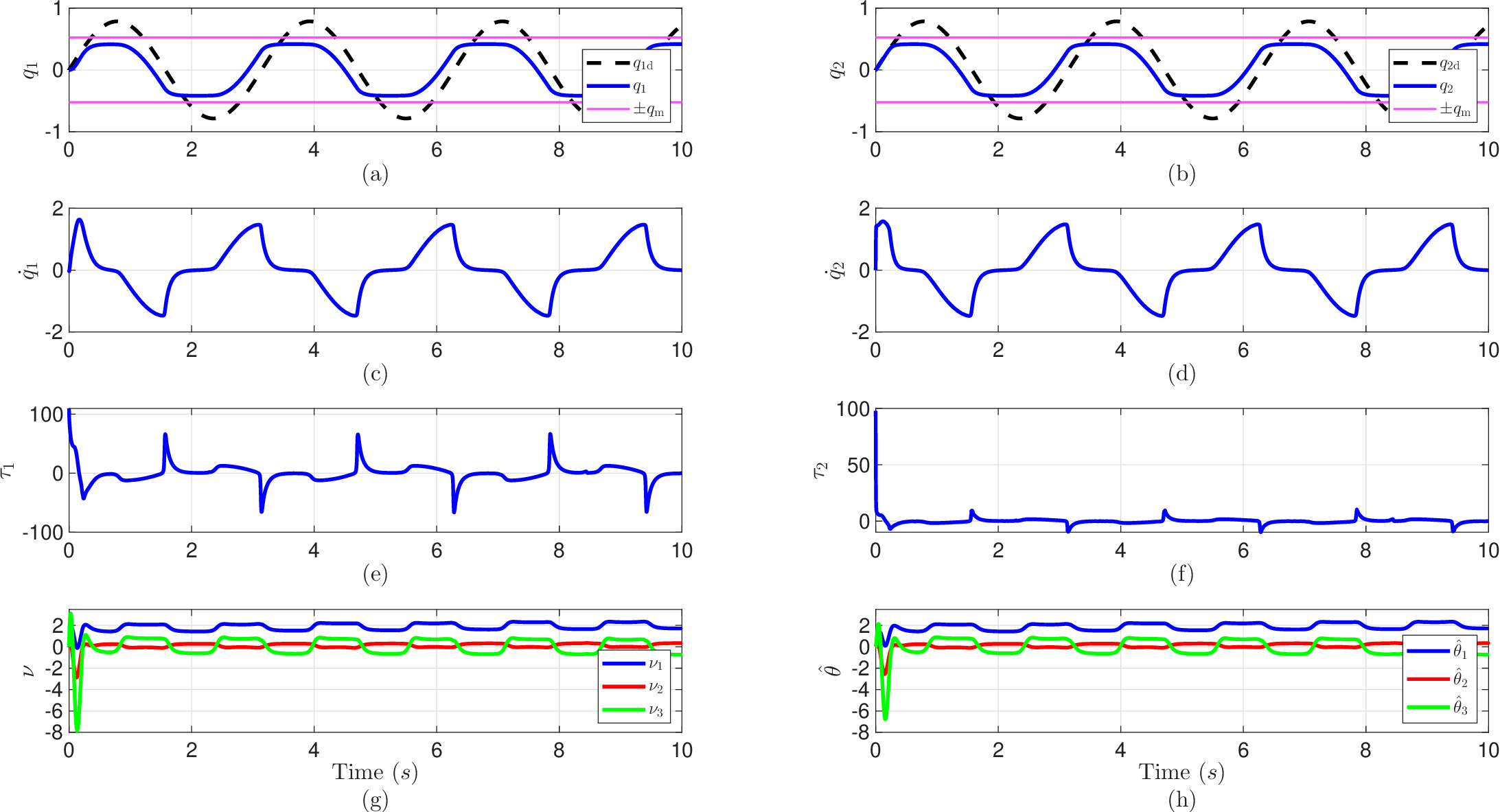}
    \vspace{-8mm}
    \caption{Adaptive safety of robotic system: (a) angle of the first link $q_1$ with reference and constraints; (b) angle of the second link $q_2$ with reference and constraints; (c)-(d) angular velocities $\dot q_1$, and $\dot q_2$; (e)-(f) torques applied to each link $\tau_1$ and $\tau_2$; (g) intermediate parameter estimates of the high order tuner $\bm \nu$; (h) parameter estimates $\hat{ \bm\theta}$.}
    \label{fig:two_link_comparison}
    \vspace{-2mm}
\end{figure}

The results demonstrate that the proposed adaptive control strategy ensures the satisfaction of configuration constraints by tracking a safe reference velocity produced by T-RaCBFs.

\section{Conclusions}
\label{sec:conclusion}
This paper presented a high-order tuner adaptive control architecture for safety-critical control that aims to mitigate high-gain conditions associated with existing approaches to adaptive safety. Our framework, including theoretical results and numerical examples, was illustrated for both systems with matched parametric uncertainty and robotic systems modeled using the manipulator equations.

\bibliographystyle{ieeetr}
\bibliography{biblio}

\end{document}